\documentclass[draft]{new_tlp}

\usepackage{amsmath}
\usepackage{mathptmx}
\usepackage[ruled,linesnumbered,vlined]{algorithm2e}
\usepackage{xspace}
\usepackage{color}
\usepackage[obeyDraft]{todonotes}\presetkeys{todonotes}{inline}{}
\usepackage{wasysym}
\usepackage{booktabs}
\usepackage{comment}
\usepackage{thm-restate}
\usepackage{enumitem}

\newcommand\bcmdtab{\noindent\bgroup\tabcolsep=0pt%
  \begin{tabular}{@{}p{10pc}@{}p{20pc}@{}}}
\newcommand\ecmdtab{\end{tabular}\egroup}

\def\naf{\ensuremath{\raise.17ex\hbox{\ensuremath{\scriptstyle\mathtt{\sim}}}}\xspace}

\DeclareMathOperator*{\argmin}{arg\,min}

\newtheorem{example}{Example}[section]

\usepackage{pgfplots}\usetikzlibrary{plotmarks}

\newcommand\pfun{\mathrel{\ooalign{\hfil$\mapstochar\mkern5mu$\hfil\cr$\to$\cr}}}

\pgfplotsset{
    filter discard warning=false 
    , legend cell align=left
    , minor grid style={loosely dotted, lightgray}
    , major grid style={loosely dashed, lightgray}
}

  \title[Shared aggregate sets in answer set programming]
        {Shared aggregate sets in answer set programming}

  \author[M. Alviano, C. Dodaro and M. Maratea]
         {
             MARIO ALVIANO\\
             DEMACS, University of Calabria, Italy\\
             \email{alviano@mat.unical.it}
             \and
             CARMINE DODARO, MARCO MARATEA\\
             DIBRIS, University of Genova, Italy\\
             \email{dodaro@dibris.unige.it, marco@dibris.unige.it}
        }

\jdate{March 2003}
\pubyear{2003}
\pagerange{\pageref{firstpage}--\pageref{lastpage}}
\doi{S1471068401001193}

\begin{document}
\label{firstpage}

\maketitle

\begin{abstract}
Aggregates are among the most frequently used linguistic extensions of answer set programming.
The result of an aggregation may introduce new constants during the instantiation of the input program, a feature known as value invention.
When the aggregation involves literals whose truth value is undefined at instantiation time, modern grounders introduce several instances of the aggregate, one for each possible interpretation of the undefined literals.
This paper introduces new data structures and techniques to handle such cases, and more in general aggregations on the same aggregate set identified in the ground program in input.
The proposed solution reduces the memory footprint of the solver without sacrificing efficiency.
On the contrary, the performance of the solver may improve thanks to the addition of some simple entailed clauses which are not easily discovered otherwise, and since redundant computation is avoided during propagation.
Empirical evidence of the potential impact of the proposed solution is given.
(Under consideration for acceptance in TPLP).
\end{abstract}


\section{Introduction}

Answer set programming (ASP) extends logic programming with linguistic constructs conceived to ease the representation of common knowledge \cite{DBLP:conf/iclp/GelfondL88,DBLP:journals/cacm/BrewkaET11,DBLP:journals/aim/JanhunenN16,DBLP:journals/aim/Lifschitz16}, and with constructs specifically designed for industrial applications \cite{DBLP:conf/aaai/GebserKROSW13,DBLP:journals/tplp/DodaroGLMRS16}.
Aggregates are among the most frequently used linguistic extensions of ASP \cite{DBLP:journals/ai/SimonsNS02,DBLP:journals/ai/LiuPST10,DBLP:conf/aaaiss/BartholomewLM11,DBLP:journals/tocl/Ferraris11,DBLP:journals/ai/FaberPL11,DBLP:journals/jair/AlvianoCFLP11,DBLP:journals/ia/Alviano11,DBLP:conf/aaai/AlvianoFS16}, providing a natural representation to reason on properties involving sets of literals.
Simple forms of aggregations combine common aggregation functions, aggregate sets, and comparison operators;
common aggregation functions include \texttt{\#sum}, \texttt{\#count} (special case of \texttt{\#sum} with uniform weights), \texttt{\#min} and \texttt{\#max} (replaced by normal rules by modern grounders).
However, an aggregation may also result into \emph{value invention}, as new constants are possibly introduced for assignments over aggregate sets involving literals whose truth value is undefined at instantiation time.
In such cases, several instances of the aggregation are produced, one for each possible interpretation of the undefined literals in the aggregate set.
Actually, such instances share the same ground aggregate set, and only the compared terms differ.

An example of the scenario described above is given by the following program (in ASP-Core-2 syntax; see \citeNP{aspcore2}):
\begin{align*}
    \begin{array}{ll}
        \tt \{p(2);\ p(5)\}. \qquad &
        \tt q(S)\ :\!-\ \ S = \#sum\{X : p(X)\}.
    \end{array}
\end{align*}
whose instantiation contains
\mbox{$\tt q(\mathit{s})\ :\!-\ \ \mathit{s} = \#sum\{2 : p(2);\ 5 : p(5)\}$}
for $s \in \{0, 2, 5, 7\}$.
Therefore, the ground program contains 4 aggregates sharing the same aggregate set.

Grounders have to apply some additional transformations, as the input language of modern solvers requires a normal form that essentially restricts rules with aggregates to the form
\mbox{$\tt \mathit{a}\ :\!-\ \ \#sum\{\mathit{w_1 : \ell_1; \cdots; w_n : \ell_n}\} \geq \mathit{b}$},
where $a$ is an atom not occurring in any other rule head, $n \geq 1$, $\ell_1,\ldots,\ell_n$ are literals, and $w_1,\ldots,w_n,b$ are positive integers.
After that, solvers such as \textsc{clasp} \cite{DBLP:journals/ai/GebserKS12,DBLP:conf/lpnmr/GebserKS09,DBLP:conf/lpnmr/GebserKK0S15,DBLP:conf/ijcai/GebserKS13} and \textsc{wasp} \cite{DBLP:conf/cilc/DodaroAFLRS11,DBLP:conf/lpnmr/AlvianoDFLR13,DBLP:conf/lpnmr/AlvianoDLR15} associate every normalized aggregate with a \emph{propagator}, that is, a specific data structure receiving notifications for truth assignments to literals $a,\ell_1,\ldots,\ell_n$, and possibly deriving new assignments for these literals.

Back to the above example, the program processed by the solver would be the following:
\begin{align*}
    \begin{array}{lllll}
        \tt \{p(2);\ p(5)\}. &
        \multicolumn{3}{l}{a_x\mathtt{\ :\!-\ \ \#sum\{2 : p(2);\ 5 : p(5)\}} {}\geq x.\quad (\forall x \in \{1,2,3,5,6,7\})} \\
        \tt q(0)\ :\!-\ \ not\ \mathit{a}_1. &
        \tt q(2)\ :\!-\ \ \mathit{a}_2,\ not\ \mathit{a}_3. &
        \tt q(5)\ :\!-\ \ \mathit{a}_5,\ not\ \mathit{a}_6. &
        \tt q(7)\ :\!-\ \ \mathit{a}_7. \\
    \end{array}
\end{align*}
where intuitively each aggregate $s = \tt \#sum\{2 : p(2);\ 5 : p(5)\}$ was replaced by the conjunction $\tt \#sum\{2 : p(2);\ 5 : p(5)\} \geq \mathit{s},\ not\ \#sum\{2 : p(2);\ 5 : p(5)\} \geq \mathit{s}+1$.
Actually, some of the aggregates in the above program differ only at a syntactic level, while they are indistinguishable at a semantic level, as they encode the same boolean function.
For example, $\tt \#sum\{2 : p(2);\ 5 : p(5)\} \geq 1$ and $\tt \#sum\{2 : p(2);\ 5 : p(5)\} \geq 2$ are both true whenever at least one literal in their aggregate set is true.
Moreover, there are some simple entailed formulas which are not easily derivable by clause learning;
specifically, $a_x \rightarrow a_y$ for any $x < y$.
Finally, notifications are always triggered to all propagators, whose computation is therefore redundant.

Aggregates with the same aggregate set are not necessarily the outcome of the instantiation of assignments, and may result by the instantiation of different rules.
Actually, solvers cannot speculate on the origin of normalized aggregates in the instantiated program, and therefore any optimization technique must rely only on the information encoded in the aggregates themselves.
The definition of such techniques is the focus of this paper.

First of all, the bound of any aggregate is replaced by the smallest possible sum being greater or equal than the bound itself, so that rule \mbox{$a_1\tt\ :\!-\ \ \#sum\{2 : p(2);\ 5 : p(5)\} \geq 1$} becomes rule \mbox{$a_1\tt\ :\!-\ \ \#sum\{2 : p(2);\ 5 : p(5)\} \geq 2$}.
After that, duplicated aggregates are removed, and aggregates with the same aggregate sets are identified, so that they can be associated with a new propagator, referred to as \emph{shared aggregate set propagator} (Section~\ref{sec:multiaggregates}).
The new propagator reduces the memory footprint of the solver, preserves the behavior of the propagators currently used by \textsc{clasp} and \textsc{wasp}, and its compact representation eases the introduction of rules encoding the entailed formulas described above.

The proposed techniques are implemented in \textsc{wasp} (Section~\ref{sec:implementation}), and tested empirically on a synthetic domain obtained by simplifying an encoding for a real world application in medical informatics (the authors were asked to analyze the inefficiency of ASP solvers on this encoding, but constrained to not disclose any sensitive information; see Section~\ref{sec:componentassignment}).
The potential impact of the new techniques is clearly evident on the synthetic domain, where the performance of \textsc{wasp} goes over the state-of-the-art solver \textsc{clasp} (Section~\ref{sec:experiment}).
An additional benchmark is obtained from instances of the ASP Competitions \cite{DBLP:conf/lpnmr/AlvianoCCDDIKKOPPRRSSSWX13,DBLP:journals/ai/CalimeriGMR16,DBLP:journals/jair/GebserMR17}, with the aim to verify the absence of overhead when shared aggregate set propagators are applied to encodings that are already handled efficiently by \textsc{clasp} and \textsc{wasp}.
Within this respect, instances with assignments over count aggregates are considered.


\section{Preliminaries}
\label{sec:prel}

This section introduces minimal background knowledge required to present the results of this paper.
Specifically, syntax and semantics of ASP programs are given, where the syntax is properly simplified to ease the presentation.
After that, the stable model search procedure implemented by modern ASP solvers is sketched, focusing on the notion of propagator.

\subsection{Syntax and semantics}\label{sec:syntax}

Let $\mathcal{A}$ be a set of \emph{atoms}.
An \emph{atomic formula} is either an atom, or the connective $\bot$.
A \emph{literal} is an atomic formula possibly preceded by the \emph{default negation} symbol \naf.
For a literal $\ell$, let $\overline{\ell}$ denote the \emph{complement} of $\ell$, that is, $\overline{p} = \naf p$ and $\overline{\naf p} = p$ for all $p \in \mathcal{A} \cup \{\bot\}$;
for a set $L$ of literals, let $\overline{L}$ be $\{\overline{\ell} \mid \ell \in L\}$.
Let $\top$ be a compact representation of $\naf\bot$.

A \emph{rule} is of one of the following forms:
\begin{align}
    \label{eq:rule}
    p_1 \vee \cdots \vee p_m \leftarrow{} & \ell_1, \ldots, \ell_n \\
    \label{eq:aggr}
    p \leftarrow{} & \textsc{sum}\{w_0 : \ell_0, \ldots, w_n : \ell_n\} \geq b
\end{align}
where $m \geq 1$, $n \geq 0$, $p_1,\ldots,p_m$ are atomic formulas, $p$ is an atom, $\ell_0,\ldots,\ell_n$ are distinct literals, and $b,w_0,\ldots,w_n$ are positive integers.
For a rule $r$ of the form (\ref{eq:rule}), let $H(r)$ denote the set $\{p_1,\ldots,p_m\} \cap \mathcal{A}$ of \emph{head atoms}, and $B(r)$ denote the set $\{\ell_1,\ldots,\ell_n\}$ of \emph{body literals} (Note that $\bot$ may occur in $B(r)$, but it is excluded from $H(r)$).
For a rule $r$ of the form (\ref{eq:aggr}), define $\mathit{id}(r) := p$, $\mathit{sum}(r) := \textsc{sum}\{w_0 : \ell_0, \ldots, w_n : \ell_n\} \geq b$, $\mathit{elem}(r) := \{(w_i,\ell_i) | i \in [0..n]\}$, and $\mathit{bound}(r) := b$.

A \emph{program} $\Pi$ is a finite set of rules.
Let $\mathit{atoms}(\Pi)$, $\mathit{rules}^\vee(\Pi)$, and $\mathit{rules}^{\sum}(\Pi)$ denote respectively the set of atoms occurring in $\Pi$, the set of rules of the form (\ref{eq:rule}) in $\Pi$, and the set of rules of the form (\ref{eq:aggr}) in $\Pi$.
In the following, every program $\Pi$ is assumed to satisfy the following property:
for each $r \in \mathit{rules}^{\sum}(\Pi)$, there is no $r' \in \mathit{rules}^\vee(\Pi)$ with $\mathit{id}(r) \in H(r')$, and there is no $r' \in \mathit{rules}^{\sum}(\Pi) \setminus \{r\}$ with $\mathit{id}(r) = \mathit{id}(r')$;
stated differently, $\mathit{id}(r)$ is an identifier for the aggregation $\mathit{sum}(r)$.

\begin{example}[Running example] \label{ex:aggregates}
    Let $\Pi_\mathit{run}$ be the following program (similar to the example in the introduction):
    \begin{align*}
    \begin{array}{lll}
    g_2:\ p_2 \vee n_2 \leftarrow{} &
    g_5:\ p_5 \vee n_5 \leftarrow{} &
    r_x:\ a_x \leftarrow \textsc{sum}\{2 : p_2, 5 : p_5\} \geq x \quad \forall x \in \{1,2,3,5,6,7\}.
    \end{array}
    \end{align*}
    Note that, for all $x \in \{1,2,3,5,6,7\}$, $\mathit{id}(r_x) = a_x$ does not occur in any other rule head, and therefore it is an identifier for $\mathit{sum}(r_x) = \textsc{sum}\{2 : p_2, 5 : p_5\} \geq x$, whose elements are $\mathit{elem}(r_x) = \{(2,p_2),$ $(5,p_5)\}$, and whose bound is $\mathit{bound}(r_x) = x$.
    \hfill $\lhd$
\end{example}

The \emph{dependency graph} $G_\Pi$ of $\Pi$ has nodes $\mathit{atoms}(\Pi)$, and an arc $xy$ (where $x$ and $y$ are atoms) for each rule $r \in \mathit{rules}^\vee(\Pi)$ such that $x \in H(r)$ and $y \in B(r)$, and for each rule $r \in \mathit{rules}^{\sum}(\Pi)$ such that $x = \mathit{id}(r)$ and $(w,y) \in \mathit{elem}(r)$ for some integer $w$.
An atom is \emph{recursive} in $\Pi$ if it is involved in a cycle of $G_\Pi$.
In the following every program $\Pi$ is assumed to satisfy the following property:
if $r \in \Pi$ is of the form (\ref{eq:aggr}), $\mathit{id}(r)$ is not recursive in $\Pi$.
Note that $\Pi_\mathit{run}$ has such a property.

An \emph{assignment} $I$ is a set of literals (different from $\bot$ and $\top$) such that $I \cap \overline{I} = \emptyset$;
literals in $I$ are true, literals in $\overline{I}$ are false, and all other literals are undefined.
An \emph{interpretation} $I$ for a program $\Pi$ is an assignment such that $I \cup \overline{I} = \mathit{atoms}(\Pi) \cup \overline{\mathit{atoms}(\Pi)}$.
Relation $\models$ is inductively defined as follows:
$I \not\models \bot$ (hence, $I \models \top$);
for $p \in \mathcal{A}$, $I \models p$ if $p \in I$;
$I \models \naf p$ if $I \not\models p$;
for a rule $r$ of the form (\ref{eq:rule}), $I \models B(r)$ if $I \models \ell$ for all $\ell \in B(r)$, $I \models H(r)$ if $I \models p$ for some $p \in H(r)$, and $I \models r$ if $I \models H(r)$ whenever $I \models B(r)$;
for a rule $r$ of the form (\ref{eq:aggr}), $I \models \mathit{sum}(r)$ if $\sum_{(w,\ell) \in \mathit{elem}(r),\ I \models \ell}{w} \geq \mathit{bound}(r)$, and $I \models r$ if $I \models \mathit{id}(r)$ whenever $I \models \mathit{sum}(r)$;
for a program $\Pi$, $I \models \Pi$ if $I \models r$ for all $r \in \Pi$.
For any expression $\pi$, if $I \models \pi$, we say that $I$ is a \emph{model} of $\pi$. 

The \emph{reduct} $\Pi^I$ of a program $\Pi$ with respect to an interpretation $I$ comprises the following rules:
for each rule $r \in \mathit{rules}^\vee(\Pi)$ such that $I \models B(r)$, $\Pi^I$ contains a rule $r^I$ of the form (\ref{eq:rule}) with $H(r^I) = H(r)$, and $B(r^I) = B(r) \cap \mathcal{A}$;
for each rule $r \in \mathit{rules}^{\sum}(\Pi)$ such that $I \models \mathit{sum}(r)$, $\Pi^I$ contains a rule $r^I$ of the form (\ref{eq:aggr}) with $\mathit{id}(r^I) = \mathit{id}(r)$, $\mathit{bound}(r^I) = \mathit{bound}(r)$, and $\mathit{elem}(r) = \{(w,\ell) \in \mathit{elem}(r) \mid I \models \ell,\ \ell \in \mathcal{A}\} \cup \{(\sum_{(w,\ell) \in \mathit{elem}(r),\ I \models \ell,\ \ell \notin \mathcal{A}}{w},\top)\}$.
An interpretation $I$ is a \emph{stable model} of a program $\Pi$ if $I \models \Pi$ and there is no $J \subset I$ such that $J \models \Pi^I$.
Let $\mathit{SM}(\Pi)$ denote the set of stable models of $\Pi$.
For instance, for $\Pi_\mathit{run}$ from Example~\ref{ex:aggregates}, $\mathit{SM}(\Pi_\mathit{run})$ comprises the following models:
$\{n_2, n_5\}$,
$\{n_2, p_5, a_1, a_2, a_3, a_5\}$,
$\{p_2, n_5, a_1, a_2\}$, and
$\{p_2, p_5, a_1, a_2, a_3, a_5, a_6, a_7\}$.

\subsection{Stable model search and propagators}\label{sec:propagators}

Stable model search is implemented in modern ASP solvers as a \emph{conflict-driven clause learning} (CDCL) algorithm \cite{DBLP:journals/ai/GebserKS12}, which is based on the pattern \textit{choose}-\textit{propagate}-\textit{learn}.
Intuitively, the idea is to build a stable model step-by-step starting from an empty assignment $I$.
At each step of computation, a branching literal is heuristically \textit{chosen} to be added in $I$, and \emph{propagated} so to add deterministic consequences to $I$, if possible.
Otherwise, if the complement of a deterministic consequence already belongs to $I$, a \emph{conflict} is identified.
Each deterministic consequence $\ell$ added to $I$ is also associated with a set of \emph{reasons}, essentially the true literals causing the addition of $\ell$ to $I$.
Conflicts are analyzed to \emph{learn} new clauses by applying backward resolution on the reasons of the conflictual literals;
while performing backward resolution, literals previously added to $I$ are removed, until the learned clause causes the addition of a new deterministic consequence that drives the search into a different branch.
This process is repeated until either $I$ is a stable model, or the empty clause is learned, meaning that the input program has no stable models.

A \emph{propagator} is a module for computing deterministic consequences of an assignment.
The simplest of such modules is \emph{unit propagation}:
unit propagation adds a literal $\ell$ to the current assignment $I$ if the input program contains a rule $r$ of the form (\ref{eq:rule}) such that $r$ can be satisfied only by $I \cup \{\ell\}$.
More formally, let $C(r) := \{p_1, \ldots, p_m, \overline{\ell_1}, \ldots, \overline{\ell_n}\}$ be the clause representation of $r$.
A literal $\ell \in C(r) \setminus I$ is unit propagated with respect to $I$ and $r$ if $\overline{C(r) \setminus \{\ell\}} \subseteq I$;
let $\mathit{reasons}(\ell)$ be $\overline{C(r) \setminus \{\ell\}}$ in this case.
For instance, atom $p_2$ is unit propagated with respect to the assignment $\{a_1, a_2, \naf n_2\}$ and rule $g_2$ from Example~\ref{ex:aggregates}, and $\mathit{reasons}(p_2)$ is $\{\naf n_2\}$.

Concerning rules of the form (\ref{eq:aggr}), specific propagators have been proposed in the literature \cite{DBLP:conf/iclp/GebserKKS09,DBLP:journals/fuin/FaberLMR11}, here referred to as \textit{aggregate propagators}.
The idea is that four types of inferences are associated with rules of the form (\ref{eq:aggr}).
In particular, given a rule $r$ of the form (\ref{eq:aggr}), and an assignment $I$, the following literals are inferred (if not already in $I$):
\begin{enumerate}
\item[(A1)]
$\mathit{id}(r)$, whenever $\sum_{(w,\ell) \in \mathit{elem}(r),\ \ell \in I}{w} \geq \mathit{bound}(r)$, with $\mathit{reasons}(\mathit{id}(r))$ comprising true literals in $\mathit{elem}(r)$, that is, $\{\ell \mid (w,\ell) \in \mathit{elem}(r),\ \ell \in I\}$;

\item[(A2)]
$\overline{\mathit{id}(r)}$, whenever $\sum_{(w,\ell) \in \mathit{elem}(r),\ \overline{\ell} \notin I}{w} < \mathit{bound}(r)$, with $\mathit{reasons}(\overline{\mathit{id}(r)})$ comprising the complements of false literals in $\mathit{elem}(r)$, that is, $\{\overline{\ell} \mid (w,\ell) \in \mathit{elem}(r),\ \overline{\ell} \in I\}$;

\item[(A3)]
any $\ell$ such that $(w,\ell) \in \mathit{elem}(r)$, whenever both $\mathit{id}(r) \in I$ and $\sum_{(w',\ell') \in \mathit{elem}(r) \setminus \{(w,\ell)\},\ \overline{\ell'} \notin I}{w'}$ ${}< \mathit{bound}(r)$, with $\mathit{reasons}(\ell)$ comprising $\mathit{id}(r)$ and the complements of false literals in $\mathit{elem}(r)$, that is, $\{\mathit{id}(r)\} \cup \{\overline{\ell'} \mid (w',\ell') \in \mathit{elem}(r),\ \overline{\ell'} \in I\}$;

\item[(A4)]
any $\ell$ such that $(w,\overline{\ell}) \in \mathit{elem}(r)$, whenever both $\overline{\mathit{id}(r)} \in I$ and $\sum_{(w',\ell') \in \mathit{elem}(r) \setminus \{(w,\overline{\ell})\},\ \ell' \in I}{w'}$ ${}\geq \mathit{bound}(r) - w$, with $\mathit{reasons}(\ell)$ comprising $\overline{\mathit{id}(r)}$ and true literals in $\mathit{elem}(r)$, that is, $\{\overline{\mathit{id}(r)}\} \cup \{\ell' \mid (w',\ell') \in \mathit{elem}(r),\ \ell' \in I\}$.
\end{enumerate}
For instance, for the rules from Example~\ref{ex:aggregates}, if $I$ is $\{a_1, a_2, \naf p_2, p_5\}$, then $a_3, a_5$ are inferred by (A1), and $\naf a_6, \naf a_7$ are inferred by (A2);
in this case, $\mathit{reasons}(a_3) = \mathit{reasons}(a_5) = \{p_5\}$, and $\mathit{reasons}(\naf a_6)$ ${}= \mathit{reasons}(\naf a_7) = \{\naf p_2\}$.
Similarly, if $I$ is $\{a_5, \naf a_6\}$, then $p_5$ is inferred by (A3), and $\naf p_2$ is inferred by (A4);
in this case, $\mathit{reasons}(p_5) = \{a_5\}$, and $\mathit{reasons}(\naf p_2) = \{\naf a_6\}$.

\section{Handling ASP programs with shared aggregate sets}\label{sec:multiaggregates}
\label{sec:aggr}

As discussed in the introduction, ASP programs may contain several rules of the form (\ref{eq:aggr}) with the same elements but different bounds.
In such cases, aggregate propagators introduced in the previous section are subject to some intrinsic inefficiency.
First of all, very similar data structures have to be stored.
Second, and more important, in order to check the applicability of the four inference rules described in the previous section, redundant computation is performed because the same sum is computed several times.
The objective of this section is to introduce a new propagator for circumventing such inefficiencies.
The new propagator actually comes with a few simplifications that are applied to the input program in order to remove trivially unsatisfiable sums, to normalize bounds, to identify some equivalent sums, and to add a few entailed integrity constraints that are not easy to discover by CDCL.

Additional notation is introduced here to simplify the presentation of the simplifications.
For a program $\Pi$, and a rule $r \in \mathit{rules}^{\sum}(\Pi)$, define $\mathit{sums}(r) := \{\sum_{(\ell,w) \in S}{w} \mid \emptyset \subset S \subseteq \mathit{elem}(r)\}$, the set of possible nonzero sums for the elements of $r$, and define $\mathit{next}(r, \Pi)$ to be 
\begin{align*}
    \argmin_{\mathit{id}(r') : r' \in \mathit{rules}^{\sum}(\Pi),\ \mathit{elem}(r') = \mathit{elem}(r),\ \mathit{bound}(r') > \mathit{bound}(r)}{\mathit{bound}(r')}
\end{align*}
if it exists, and $\bot$ otherwise;
that is, $\mathit{next}(r,\Pi)$ is the id of the rule with the next possible sum (if it exists; if there are several, one is selected according to some order, as for example lexicographical order).
The following rewritings, from the first to the last, are applied to every rule $r \in \mathit{rules}^{\sum}(\Pi)$:
\begin{itemize}
	\item[(i)] if $\mathit{bound}(r) > \sum_{(w,\ell) \in \mathit{elem}(r)} w$ then $\Pi := (\Pi \setminus \{r\}) \cup \{\bot \leftarrow \mathit{id}(r)\}$;
	
	\item[(ii)] $\mathit{bound}(r) := \min(\{b \in \mathit{sums}(r) \mid b \geq \mathit{bound}(r)\})$;
	
	\item[(iii)] if there is $r' \in \mathit{rules}^{\sum}(\Pi)$ such that $\mathit{sum}(r') = \mathit{sum}(r)$, $\Pi := (\Pi \setminus \{r\}) \cup \{\mathit{id}(r) \leftarrow \mathit{id}(r')\}$;	
	
	\item[(iv)] $\Pi := \Pi \cup \{\bot \leftarrow \naf \mathit{id}(r), next(r, \Pi)\}$.
\end{itemize}
Intuitively, rewriting (i) eliminates from the input program all trivially unsatisfiable sums, that is, those which are unsatisfied even when all their elements are true.
Rewriting (ii) normalizes bounds of sums, so that the bound of each rule can be obtained by summing some weights occurring in the rule.
After that, rewriting (iii) identifies rules with the same sum, and leave a single copy of the aggregation, while identifiers of removed sums are forced to preserve their semantics by a new, simple rule.
Finally, rewriting (iv) introduces entailed integrity constraints that are difficult to identify otherwise.
Note that rewritings (i), (ii) and (iii) can be applied, in this order, one rule at time, while rewriting (iv) has to be delayed until the other rewritings have been applied to all rules of the program.

\begin{restatable}{theorem}{ThmRew}\label{thm:rew}
Let $\Pi$ be a program, and $\Pi'$ be obtained by applying (i)--(iv) to $\Pi$.
$\mathit{SM}(\Pi) = \mathit{SM}(\Pi')$ holds.
\end{restatable}

\begin{proof}
	Let $\Pi_0,\ldots,\Pi_n$ ($n \geq 0$) be the sequence of programs obtained by applying (i)--(iv), that is, $\Pi_0 = \Pi$, $\Pi_n = \Pi'$, and for all $j \geq 0$ $\Pi_{j+1}$ is obtained from $\Pi_j$ by a single application of a rewriting among (i)--(iv).
	We shall show that $\mathit{SM}(\Pi_{j}) = \mathit{SM}(\Pi_{j+1})$, for any $j \geq 0$.
	
	\begin{enumerate}[leftmargin=*,labelsep=2pt,label=(\Alph{enumi})]
		\item[(i)]
		Let $\Pi_{j+1}$ be obtained by applying (i) to $\Pi_j$.
		\begin{enumerate}[leftmargin=*,labelsep=2pt,label=(\Alph{enumi})]
			\item[($\subseteq$)]
			Let $I \in \mathit{SM}(\Pi_j)$.
			Since $\mathit{bound}(r) > \sum_{(w,\ell) \in \mathit{elem}(r)} w$, $I \not\models \mathit{sum}(r)$ holds, and since $\mathit{id}(r)$ does not occur in any other rule head, $I \not\models \mathit{id}(r)$ holds as well (otherwise, $I$ could not be a stable model).
			Therefore, $I \models \bot \leftarrow \mathit{id}(r)$, and $I \models \Pi_{j+1}$.
			Moreover, $\Pi_j^I = \Pi_{j+1}^I$, and therefore $I \in \mathit{SM}(\Pi_{j+1})$.
			
			\item[($\supseteq$)]
			Let $I \in \mathit{SM}(\Pi_{j+1})$.
			Since $\bot \leftarrow \mathit{id}(r)$ belongs to $\Pi_{j+1}$, $I \not\models \mathit{id}(r)$, and therefore $I \models \Pi_j$.
			Since $\mathit{bound}(r) > \sum_{(w,\ell) \in \mathit{elem}(r)} w$, $I \not\models \mathit{sum}(r)$ holds, and therefore $\Pi_j^I = \Pi_{j+1}^I$.
			Hence, $I \in \mathit{SM}(\Pi_j)$.
		\end{enumerate}
		
		\item[(ii)]
		Let $\Pi_{j+1}$ be obtained by applying (ii) to $\Pi_j$.
		For any interpretation $I$, $I \models \mathit{sum}(r)$ if and only if $I \models \textsc{sum}(\mathit{elem}(r)) \geq \min(\{b \in \mathit{sums}(r) \mid b \geq \mathit{bound}(r)\}$, that is, the two aggregates define the same boolean function.
		
		\item[(iii)]
		Let $\Pi_{j+1}$ be obtained by applying (iii) to $\Pi_j$.
		\begin{enumerate}[leftmargin=*,labelsep=2pt,label=(\Alph{enumi})]
			\item[($\subseteq$)]
			Let $I \in \mathit{SM}(\Pi_j)$.
			If $I \not\models \mathit{sum}(r)$, then $I \not\models \mathit{sum}(r')$, $I \not\models \mathit{id}(r)$, and $I \not\models \mathit{id}(r')$;
			hence, $I \models \Pi_{j+1}$, and $\Pi_j^I = \Pi_{j+1}^I$, which imply $I \in \mathit{SM}(\Pi_{j+1})$.
			
			Otherwise, if $I \models \mathit{sum}(r)$, then $I \models \mathit{sum}(r')$, $I \models \mathit{id}(r)$, and $I \models \mathit{id}(r')$.
			Hence, $I \models \Pi_{j+1}$ holds.
			It remains to show that $J \subset I$ implies $J \not\models \Pi_{j+1}^I$.
			Since $I \in \mathit{SM}(\Pi_j)$ by assumption, $J \not\models \Pi_j^I$ holds.
			Let $J$ be such that $J \models \Pi_j^I \cap \Pi_{j+1}^I$ (otherwise it is trivial);
			hence, there is $r'' \in \Pi_j^I \setminus \Pi_{j+1}^I$ such that $J \not\models r''$, and such a rule has been obtained from $r$.
			Thus, $J \models \mathit{sum}(r)$, and $J \not\models \mathit{id}(r)$.
			In this case, $J \not\models \mathit{id}(r) \leftarrow \mathit{id}(r')$ because $I \models \mathit{id}(r')$ follows from $J \models \Pi_j^I \cap \Pi_{j+1}^I$ and $J \models \mathit{sum}(r)$.
			It turns out that $I \in \mathit{SM}(\Pi_{j+1})$ holds.
			
			\item[($\supseteq$)]
			Let $I \in \mathit{SM}(\Pi_{j+1})$.
			If $I \not\models \mathit{sum}(r)$, then $I \not\models \mathit{sum}(r')$, $I \not\models \mathit{id}(r)$, and $I \not\models \mathit{id}(r')$;
			hence, $I \models \Pi_{j}$, and $\Pi_j^I = \Pi_{j+1}^I$, which imply $I \in \mathit{SM}(\Pi_{j+1})$.
			
			Otherwise, if $I \models \mathit{sum}(r)$, then $I \models \mathit{sum}(r')$, $I \models \mathit{id}(r)$, and $I \models \mathit{id}(r')$.
			Hence, $I \models \Pi_{j}$ holds.
			It remains to show that $J \subset I$ implies $J \not\models \Pi_{j}^I$.
			Since $I \in \mathit{SM}(\Pi_{j+1})$ by assumption, $J \not\models \Pi_{j+1}^I$ holds.
			Let $J$ be such that $J \models \Pi_j^I \cap \Pi_{j+1}^I$ (otherwise it is trivial);
			hence, there is $r'' \in \Pi_{j+1}^I \setminus \Pi_{j}^I$ such that $J \not\models r''$, and such a rule is $\mathit{id}(r) \leftarrow \mathit{id}(r')$, that is, $J \not\models \mathit{id}(r)$ and $J \models \mathit{id}(r')$.
			If $J \models \mathit{sum}(r)$, then $J \not\models r$, and therefore $J \not\models \Pi_j^I$.
			Otherwise, if $J \not\models \mathit{sum}(r)$, then $J \setminus \{\mathit{id}(r')\} \models \Pi_{j+1}^I$, which would contradict the assumption that $I \in \mathit{SM}(\Pi_{j+1})$.
			It turns out that $I \in \mathit{SM}(\Pi_j)$ holds.
		\end{enumerate}
		
		\item[(iv)]
		Let $\Pi_{j+1}$ be obtained by applying (iv) to $\Pi_j$.
		\begin{enumerate}[leftmargin=*,labelsep=2pt,label=(\Alph{enumi})]
			\item[($\subseteq$)]
			Let $I \in \mathit{SM}(\Pi_j)$.
			If $I \models \mathit{id}(r)$, then $I \models \Pi_{j+1}$ and $\Pi_j^I = \Pi_{j+1}^I$, so $I \in \mathit{SM}(\Pi_{j+1})$.
			Otherwise, let $I \not\models \mathit{id}(r)$.
			Therefore, $I \not\models \textsc{sum}(\mathit{elem}(r)) \geq b$ for all $b \geq \mathit{bound}(r)$.
			Hence, $I \not\models \mathit{next}(r,\Pi)$, and therefore $I \models \bot \leftarrow \naf \mathit{id}(r), \mathit{next}(r,\Pi)$.
			It turns out that $I \models \Pi_{j+1}$ and $\Pi_j^I = \Pi_{j+1}^I$, so $I \in \mathit{SM}(\Pi_{j+1})$.
			
			\item[($\supseteq$)]
			Let $I \in \mathit{SM}(\Pi_{j+1})$.
			Since $\Pi_j \subset \Pi_{j+1}$, $I \models \Pi_j$.
			Moreover, $\Pi_j^I = \Pi_{j+1}^I$, and therefore $I \in \mathit{SM}(\Pi_j)$.
		\end{enumerate}
	\end{enumerate}
	The proof is complete.
	\hfill
\end{proof}

\begin{example}[Continuing Example~\ref{ex:aggregates}]\label{ex:simplifications}
Note that, for all $x \in \{1,2,3,5,6,7\}$, $\mathit{sums}(r_x) = \{2, 5, 7\}$.
The application of (i)--(iv) produces the following program $\Pi_\mathit{run}'$:
\begin{align*}
    \begin{array}{lll}
    g_2: \quad p_2 \vee n_2 \leftarrow{} &
    g_5: \quad p_5 \vee n_5 \leftarrow{} \\
    r_1': \quad a_1 \leftarrow a_2 &
    r_2: \quad a_2 \leftarrow \textsc{sum}\{2 : p_2, 5 : p_5\} \geq 2 &
    s_2: \quad \bot \leftarrow \naf a_2, a_5 \\
    r_3': \quad a_3 \leftarrow a_5 & 
    r_5: \quad a_5 \leftarrow \textsc{sum}\{2 : p_2, 5 : p_5\} \geq 5 &
    s_5: \quad \bot \leftarrow \naf a_5, a_7 \\
    r_6': \quad a_6 \leftarrow a_7 &
    r_7: \quad a_7 \leftarrow \textsc{sum}\{2 : p_2, 5 : p_5\} \geq 7 &
    s_7: \quad \bot \leftarrow \naf a_7, \bot
    \end{array}
\end{align*}
For instance, the application of (ii) modifies the bound of $r_1$ from 1 to 2, so that it now occurs in $\mathit{sums}(r_1)$;
the resulting rule is then processed by (iii), and replaced by $r_1'$.
Rules $s_2$, $s_5$, and $s_7$ are added by the application of (iv).
\hfill $\lhd$
\end{example}

Shared aggregate sets are identified during the simplification process described above.
For each of these shared aggregate sets, a specific propagator is instantiated.
The new propagator, referred to as \emph{shared aggregate set propagator}, has data structures specifically conceived to implement the inference rules of several aggregate propagators avoiding redundant computation.
More formally, for a program $\Pi$, let $X$ be a maximal subset of $\mathit{rules}^{\sum}(\Pi)$ such that all $r,r' \in X$ satisfy $\mathit{elem}(r) = \mathit{elem}(r')$.
A shared aggregate set propagator for $X$ is associated with the following sets:
$\mathit{elem}(X)$, storing the shared set of elements, that is, $\mathit{elem}(r)$, for any $r \in X$; and
$\mathit{bounds}(X)$, storing identifiers and bounds, that is, $\{(\mathit{bound}(r),\mathit{id}(r)) \mid r \in X\}$.
Moreover, the following functions are used by the propagator when applied to an assignment $I$:
\begin{itemize}
\item
$\mathrm{min\_sum}(X,I) := \sum_{(w,\ell) \in \mathit{elem}(X),\ \ell \in I} w$, the smallest value of the sum among those assigned by interpretations extending the current assignment;

\item
$\mathrm{max\_sum}(X,I) := \sum_{(w,\ell) \in \mathit{elem}(X),\ \overline{\ell} \notin I} w$, the greatest value of the sum among those assigned by interpretations extending of the current assignment;

\item
$\mathrm{lower\_bound}(X,I) := \max(\{b \mid (b,\ell) \in \mathit{bounds}(X),\ \ell \in I\} \cup \{0\})$, a value that must be reached by the sum in any extension of the current assignment;

\item
$\mathrm{upper\_bound}(X,I) := \min(\{b \mid (b,\ell) \in \mathit{bounds}(X),\ \overline{\ell} \in I\} \cup \{+\infty\})$, a value that must not be reached by the sum in any extension of the current assignment.
\end{itemize}

\begin{example}[Continuing Example~\ref{ex:simplifications}]\label{ex:propagatorsets}
Let $X$ be $\mathit{rules}^{\sum}(\Pi_\mathit{run}'))$.
The shared aggregate set propagator associated with $X$ has $\mathit{elem}(X) = \{(2,p_2), (5,p_5)\}$, and $\mathit{bounds}(X) = \{(2,a_2), (5,a_5), (7,a_7)\}$.
For the empty assignment, \linebreak
$\mathrm{min\_sum}(X,\emptyset) = 0$,
$\mathrm{max\_sum}(X,\emptyset) = 7$,
$\mathrm{lower\_bound}(X,\emptyset) = 0$, and
$\mathrm{upper\_bound}(X,\emptyset) = +\infty$;
for $I = \{p_2,\naf p_5,a_5,$ $\naf a_7\}$,
$\mathrm{min\_sum}(X,I) = 2$,
$\mathrm{max\_sum}(X,I) = 2$,
$\mathrm{lower\_bound}(X,I) = 5$, and
$\mathrm{upper\_bound}(X,I) = 7$.
$\hfill \lhd$
\end{example}

\begin{function}[t]
    \caption{OnLiteralTrue($X$, $I$, $\ell$)}\label{fn:onliteraltrue}
    \If{$(w,\ell) \in \mathit{elem}(X)$}{\label{alg:ln:litrue}
        $R := \{\ell' \mid (w',\ell') \in \mathit{elem}(X), \ \ell' \in I\}$\;
        $R' :=  R \cup \{ \overline{\ell'} \mid (\mathrm{upper\_bound}(X,I), \ell') \in \mathit{bounds}(X)\}$\;
        \Return $\{(\ell', R) \mid (b,\ell') \in \mathit{bounds}(X), \ \ell' \notin I, \ b \leq \mathrm{min\_sum}(X,I)\} \cup {}$ \phantom{xxxx} $\{(\overline{\ell'}, R') \mid (w',\ell') \in \mathit{elem}(X), \ \overline{\ell'} \notin I, \ w' \geq \mathrm{upper\_bound}(X,I) - \mathrm{min\_sum}(X,I)\}$\;
    }
    \If{$(w,\overline{\ell}) \in \mathit{elem}(X)$}{\label{alg:ln:litfalse}
        $R := \{\overline{\ell'} \mid (w',\ell') \in \mathit{elem}(X), \ \overline{\ell'} \in I\}$\;
        $R' := R \cup \{ \ell' \mid (\mathrm{lower\_bound}(X,I),\ell') \in \mathit{bounds}(X)\}$\;
        \Return $\{(\overline{\ell'},R) \mid (b,\ell') \in \mathit{bounds}(X), \ \overline{\ell'} \notin I, \ b > \mathrm{max\_sum}(X,I)\} \cup {}$ \phantom{xxxx} $\{(\ell', R') \mid (w',\ell') \in \mathit{elem}(X), \ \ell' \not \in I, \ w' > \mathrm{max\_sum}(X,I) - \mathrm{lower\_bound}(X,I)\}$\;
    }
    \If{$(b,\ell) \in \mathit{bounds}(X)$ $\mathbf{and}$}{\label{alg:ln:idtrue}  
            $R := \{\overline{\ell'} \mid (w',\ell') \in \mathit{elem}(X), \ \overline{\ell'} \in I\} \cup \{\ell\}$\;
            \Return $\{(\ell', R) \mid (w',\ell') \in \mathit{elem}(X), \ \ell' \notin I, \ w' > \mathrm{max\_sum}(X,I) - \mathrm{lower\_bound}(X,I)\}$\;
    }
    \If{$(b,\overline{\ell}) \in \mathit{bounds}(X)$ $\mathbf{and}$}{\label{alg:ln:idfalse} 
            $R := \{\ell' \mid (w',\ell') \in \mathit{elem}(X), \ \ell' \in I\} \cup \{\ell\}$\;
            \Return $\{(\overline{\ell'}, R) \mid (w',\ell') \in \mathit{elem}(X), \ \overline{\ell'} \notin I, \ w' \geq \mathrm{upper\_bound}(X,I) - \mathrm{min\_sum}(X,I)\}$\;
    }
    \Return $\emptyset$\;
\end{function}

All inferences of shared aggregate set propagators are shown as function \ref{fn:onliteraltrue}, whose input is the set $X$ of rules handled by the propagator, an assignment $I$, and a literal $\ell$, where $\ell$ is the last literal added to $I$.
The output of \ref{fn:onliteraltrue}($X$,$I$,$\ell$) is a set of pairs of the form $(\ell', R)$, where $\ell'$ is an inferred literal, and $R$ a set of literals;
the solver has to add $\ell'$ to $I$, and to assign $R$ to $\mathit{reasons}(\ell')$.
Specifically, \ref{fn:onliteraltrue}($X$,$I$,$\ell$) may infer something in the following cases:

\smallskip
\noindent
\emph{Case 1.}
If $\ell$, the last literal added to $I$, occurs in an element $(w,\ell)$ of the aggregate set (line \ref{alg:ln:litrue}), then the minimum value that can be assigned to the sum is augmented by $w$, and therefore identifiers associated with bounds being less or equal than such a value can be inferred;
the reasons of the inferred literals are the true literals occurring in the aggregate set (line~2);
such an inference is analogous to (A1) for aggregate propagators.
Moreover, every literal in the aggregate set whose truth would lead to a conflict are inferred false, where the conflict would arise if the addition of the weight of the inferred literal exceeds the currently known upper bound;
the reasons of the inferred literals are the true literals in the aggregate set and the complement of the literal associated to the currently known upper bound (line~3);
such an inference is analogous to (A4) for aggregate propagators.

\smallskip
\noindent
\emph{Case 2.}
If the complement of $\ell$ occurs in an element $(w,\overline{\ell})$ of the aggregate set (line \ref{alg:ln:litfalse}), then the maximum value that can be assigned to the sum is decremented by $w$, and therefore identifiers associated with bounds being greater than such a value can be inferred false;
the reasons of the inferred literals are the complements of the false literals in the aggregate set (line~6);
such an inference is analogous to (A2) for aggregate propagators.
Moreover, in this case every literal in the aggregate set whose falsity would lead to a conflict are inferred true, where the conflict would arise if the subtraction of the weight of the inferred literal is under the currently known lower bound;
the reasons of the inferred literals are the complements of the false literals in the aggregate set and the literal associated to the currently known upper bound (line~7);
such an inference is analogous to (A3) for aggregate propagators.

\smallskip
\noindent
\emph{Case 3.}
If $\ell$ occurs in an element $(b,\ell)$ of $\mathit{bounds}(X)$ (line \ref{alg:ln:idtrue}), then the lower bound is possibly increased, and therefore every literal in the aggregate set whose falsity would lead to a conflict are inferred true, where the conflict would arise if the subtraction of the weight of the inferred literal does not reach the currently known lower bound;
the reasons of the inferred literals are the complements of the false literals in the aggregate set and the literal $\ell$ (line~10);
such an inference is analogous to (A3) for aggregate propagators.

\smallskip
\noindent
\emph{Case 4.}
If the complement of $\ell$ occurs in an element $(b,\overline{\ell})$ of $\mathit{bounds}(X)$ (line \ref{alg:ln:idfalse}), then the upper bound is possibly decreased, and therefore every literal in the aggregate set whose truth would lead to a conflict are inferred false, where the conflict would arise if the addition of the weight of the inferred literal exceeds the currently known upper bound;
the reasons of the inferred literals are the true literals in the aggregate set and the literal $\ell$ (line~13);
such an inference is analogous to (A4) for aggregate propagators.

\begin{restatable}{theorem}{ThmInferences}\label{thm:inferences}
Let $\Pi$ be a program obtained by applying rewritings (i)--(iv), $X$ be a maximal subset of $\mathit{rules}^{\sum}(\Pi)$ such that $r,r' \in X$ satisfy $\mathit{elem}(r) = \mathit{elem}(r')$, $I$ be an assignment, and $\ell$ be the last literal added to $I$.
A literal $\ell'$ is inferred by $\mathrm{OnLiteralTrue}(X,I,\ell)$ with reasons $R$ if and only if $\ell'$ is inferred with reasons $R$ by applying (A1)--(A4) to some $r \in X$ and assignment $I$.
\end{restatable}

\begin{proof}
	Let $\ell'$ be inferred by $\mathrm{OnLiteralTrue}(X,I,\ell)$.
	Hence, $\ell' \notin I$ holds.
	Four cases are possible.
	\begin{enumerate}[leftmargin=*,labelsep=2pt]
		\item
		If there is $(w,\ell) \in \mathit{elem}(X)$, we further distinguish two cases:
		\begin{itemize}[leftmargin=*,labelsep=2pt]
			\item
			There is $(b,\ell') \in \mathit{bounds}(X)$ such that $b \leq \mathrm{min\_sum}(X,I)$.
			Hence, there is $r \in X$ such that $\mathit{bound}(r) = b$, and $\mathit{id}(r) = \ell'$.
			Since $b \leq \mathrm{min\_sum}(X,I)$, (A1) derives $\mathit{id}(r)$.
			Moreover, both propagators set $\mathit{reasons}(\ell') = \{\ell'' \mid (w'',\ell'') \in \mathit{elem}(X), \ell'' \in I\}$.
			
			\item
			There is $(w',\overline{\ell'}) \in \mathit{elem}(X)$ such that $w' \geq \mathrm{upper\_bound}(X,I) - \mathrm{min\_sum}(X,I)$.
			Hence, there is $r \in X$ such that $\mathit{bound}(r) = \min(\{b \mid (b,\ell'') \in \mathit{bounds}(X), \overline{\ell''} \in I\})$.
			Thus, $\overline{\mathit{id}(r)} \in I$, and (A4) derives $\ell'$.
			Moreover, both propagators set $\mathit{reasons}(\ell') = \{\overline{\mathit{id}(r)}\} \cup \{\ell'' \mid (w'',\ell'') \in \mathit{elem}(X), \ell'' \in I\}$.
		\end{itemize}
		
		\item
		If there is $(w,\overline{\ell}) \in \mathit{elem}(X)$, we further distinguish two cases:
		\begin{itemize}[leftmargin=*,labelsep=2pt]
			\item
			There is $(b,\overline{\ell'}) \in \mathit{bounds}(X)$ such that $b > \mathrm{max\_sum}(X,I)$.
			Hence, there is $r \in X$ such that $\mathit{bound}(r) = b$, and $\mathit{id}(r) = \overline{\ell'}$.
			Since $b > \mathrm{max\_sum}(X,I)$, (A2) derives $\overline{\mathit{id}(r)} = \ell'$.
			Moreover, both propagators set $\mathit{reasons}(\ell') = \{\overline{\ell''} \mid (w'',\ell'') \in \mathit{elem}(X), \overline{\ell''} \in I\}$.
			
			\item
			There is $(w',\ell') \in \mathit{elem}(X)$ such that $w' > \mathrm{max\_sum}(X,I) - \mathrm{lower\_bound}(X,I)$.
			Hence, there is $r \in X$ such that $\mathit{bound}(r) = \max(\{b \mid (b,\ell'') \in \mathit{bounds}(X), \ell'' \in I\})$.
			Thus, $\mathit{id}(r) \in I$, and (A3) derives $\ell'$.
			Moreover, both propagators set $\mathit{reasons}(\ell') = \{\mathit{id}(r)\} \cup \{\overline{\ell''} \mid (w'',\ell'') \in \mathit{elem}(X), \overline{\ell''} \in I\}$.
		\end{itemize}
		
		\item
		If there is $(b,\ell) \in \mathit{bounds}(X)$, then there is $(w',\ell') \in \mathit{elem}(X)$ such that $w' > \mathrm{max\_sum}(X,I) - \mathrm{lower\_bound}(X,I)$.
		Hence, there is $r \in X$ such that $\mathit{bound}(r) = \max(\{b \mid (b,\ell'') \in \mathit{bounds}(X),$ $\ell'' \in I\})$.
		Thus, $\mathit{id}(r) \in I$, and (A3) derives $\ell'$.
		Moreover, both propagators set $\mathit{reasons}(\ell') = \{\mathit{id}(r)\} \cup \{\overline{\ell''} \mid (w'',\ell'') \in \mathit{elem}(X), \overline{\ell''} \in I\}$.
		
		\item
		If there is $(b,\overline{\ell}) \in \mathit{bounds}(X)$, there is $(w',\overline{\ell}) \in \mathit{elem}(X)$ such that $w' \geq \mathrm{upper\_bound}(X,I) - \mathrm{min\_sum}(X,I)$.
		Hence, there is $r \in X$ such that $\mathit{bound}(r) = \min(\{b \mid (b,\ell'') \in \mathit{bounds}(X), \overline{\ell''} \in I\})$.
		Thus, $\overline{\mathit{id}(r)} \in I$, and (A4) derives $\ell'$.
		Moreover, both propagators set $\mathit{reasons}(\ell')$ equals to $\{\overline{\mathit{id}(r)}\} \cup \{\ell'' \mid (w'',\ell'') \in \mathit{elem}(X), \ell'' \in I\}$.
	\end{enumerate}
	
	As for the other direction, let $\ell'$ be inferred by (A1)--(A4) applied on $I$, and not derivable from $I \setminus \{\ell\}$.
	We distinguish four cases.
	\begin{enumerate}[leftmargin=*,labelsep=2pt]
		\item
		Literal $\ell'$ is inferred by (A1).
		Hence, there is $r \in X$ such that $\mathit{id}(r) = \ell'$, $\sum_{(w'',\ell'') \in \mathit{elem}(r), \ell'' \in I}{w''} \geq \mathit{bound}(r)$, and $\sum_{(w'',\ell'') \in \mathit{elem}(r), \ell'' \in I \setminus \{\ell\}}{w''} < \mathit{bound}(r)$.
		Thus, there is $(w,\ell) \in \mathit{elem}(X)$, and $\ell'$ is derived by $\mathrm{OnLiteralTrue}(X,I,\ell)$ at line~4 because $(\mathit{bound}(r),\ell') \in \mathit{bounds}(X)$ and $\mathit{bound}(r) \leq \mathrm{min\_sum}(X,I)$.
		Moreover, both propagators set $\mathit{reasons}(\ell') = \{\ell'' \mid (w'',\ell'') \in \mathit{elem}(X), \ell'' \in I\}$.
		
		\item
		Literal $\ell'$ is inferred by (A2).
		Hence, there is $r \in X$ such that $\overline{\mathit{id}(r)} = \ell'$, $\sum_{(w'',\ell'') \in \mathit{elem}(r), \overline{\ell''} \notin I}{w''} < \mathit{bound}(r)$, and $\sum_{(w'',\ell'') \in \mathit{elem}(r), \overline{\ell''} \notin I \setminus \{\ell\}}{w''} \geq \mathit{bound}(r)$.
		Thus, there is $(w,\overline{\ell}) \in \mathit{elem}(X)$, and $\ell'$ is derived by $\mathrm{OnLiteralTrue}(X,I,\ell)$ at line~8 because $(\mathit{bound}(r),\overline{\ell'}) \in \mathit{bounds}(X)$ and $\mathit{bound}(r) > \mathrm{max\_sum}(X,I)$.
		Moreover, both propagators set $\mathit{reasons}(\ell')$ equals to $\{\mathit{id}(r)\} \cup \{\overline{\ell''} \mid (w'',\ell'') \in \mathit{elem}(X), \overline{\ell''} \in I\}$.
		
		\item
		Literal $\ell'$ is inferred by (A3).
		Hence, there is $r \in X$ such that $(w',\ell') \in \mathit{elem}(r)$, $\mathit{id}(r) \in I$, $\sum_{(w'',\ell'') \in \mathit{elem}(r) \setminus \{(w',\ell')\}, \overline{\ell''} \notin I}{w''} < \mathit{bound}(r)$, and either $\mathit{id}(r) \notin I \setminus \{\ell\}$ or the following inequality holds:
		$\sum_{(w'',\ell'') \in \mathit{elem}(r) \setminus \{(w',\ell')\}, \overline{\ell''} \notin I \setminus \{\ell\}}{w''} \geq \mathit{bound}(r)$.
		\begin{itemize}
			\item
			If $\mathit{id}(r) \notin I \setminus \{\ell\}$, then $\mathit{id}(r) = \ell$, and therefore $\ell'$ is derived by $\mathrm{OnLiteralTrue}(X,I,\ell)$ at line~11 because $(\mathit{bound}(r),\ell) \in \mathit{bounds}(X)$, $(w',\ell') \in \mathit{elem}(X)$ and $w'$ is greater than $\mathrm{max\_sum}(X,I) - \mathrm{lower\_bound}(X,I)$.
			Moreover, both propagators set $\mathit{reasons}(\ell')$ equals to $\{\mathit{id}(r)\} \cup \{\overline{\ell''} \mid (w'',\ell'') \in \mathit{elem}(X), \overline{\ell''} \in I\}$.
			
			\item
			If $\sum_{(w'',\ell'') \in \mathit{elem}(r) \setminus \{(w',\ell')\}, \overline{\ell''} \notin I \setminus \{\ell\}}{w''} \geq \mathit{bound}(r)$, there is $(w,\overline{\ell}) \in \mathit{elem}(X)$.
			Hence, $\ell'$ is derived by $\mathrm{OnLiteralTrue}(X,I,\ell)$ at line~8 because there is $(w',\ell') \in \mathit{elem}(X)$ such that $w' > \mathrm{max\_sum}(X,I) - \mathrm{lower\_bound}(X,I)$.
			Moreover, both propagators set $\mathit{reasons}(\ell')$ equals to $\{\mathit{id}(r)\} \cup \{\overline{\ell''} \mid (w'',\ell'') \in \mathit{elem}(X), \overline{\ell''} \in I\}$.
		\end{itemize}
		
		\item
		Literal $\ell'$ is inferred by (A4).
		Hence, there is $r \in X$ such that $(w',\overline{\ell'}) \in \mathit{elem}(r)$, $\overline{\mathit{id}(r)} \in I$, $\sum_{(w'',\ell'') \in \mathit{elem}(r) \setminus \{(w',\overline{\ell'})\}, \ell'' \in I}{w''} \geq \mathit{bound}(r) - w'$, and either $\overline{\mathit{id}(r)} \notin I \setminus \{\ell\}$ or the inequality $\sum_{(w'',\ell'') \in \mathit{elem}(r) \setminus \{(w',\overline{\ell'})\}, \ell'' \in I \setminus \{\ell\}}{w''} < \mathit{bound}(r) - w'$ holds.
		\begin{itemize}
			\item
			If $\overline{\mathit{id}(r)} \notin I \setminus \{\ell\}$, then $\overline{\mathit{id}(r)} = \ell$, and therefore $\ell'$ is derived by $\mathrm{OnLiteralTrue}(X,I,\ell)$ at line~14 because there is $(w',\overline{\ell'}) \in \mathit{elem}(X)$ such that the following inequality holds:
			$w' \geq \mathrm{upper\_bound}(X,I) - \mathrm{min\_sum}(X,I)$.
			Moreover, both propagators set $\mathit{reasons}(\ell')$ equals to $\{\overline{\mathit{id}(r)}\} \cup \{\ell'' \mid (w'',\ell'') \in \mathit{elem}(X), \ell'' \in I\}$.
		\end{itemize}
	\end{enumerate}
	The proof is complete.
	\hfill
\end{proof}

\begin{example}[Continuing Example~\ref{ex:propagatorsets}]
Let $X$ be $\mathit{rules}^{\sum}(\Pi_\mathit{run}')$, and recall that $\mathit{elem}(X) = \{(2,p_2), (5,p_5)\}$, and $\mathit{bounds}(X) = \{(2,a_2),
$ $(5,a_5), (7,a_7)\}$.
$\mathrm{OnLiteralTrue}(X, \{\naf a_7, \naf n_2, p_2\}, p_2)$ returns 
$(a_2, \{p_2\})$ and $(\naf p_5, \{p_2, \naf a_7\})$ because of case~1;
indeed, literal $a_2$ is inferred because $\mathrm{min\_sum}(\{\naf a_7, \naf n_2, p_2\}) = 2$, and literal $\naf p_5$ is inferred because $\mathrm{upper\_bound}(\{\naf a_7, \naf n_2, p_2\}) - \mathrm{min\_sum}(\{\naf a_7, \naf n_2, p_2\}) = 5$.
Similarly, $\mathrm{OnLiteralTrue}(X, \{a_2, n_5, \naf p_5\}, \naf p_5)$ returns $(\naf a_5, \{\naf p_5\})$, $(\naf a_7, \{\naf p_5\})$ and $(p_2, \{\naf p_5, a_2\})$ because of case~2.
$\mathrm{OnLiteralTrue}(X, \{a_5\}, a_5)$ returns the pair $(p_5, \{a_5\})$ because of case~3.
Finally, $\mathrm{OnLiteralTrue}(X, \{\naf a_5\}, \naf a_5)$ returns the pair $(\naf p_5, \{\naf a_5\})$ because of case~4.
$\hfill \lhd$
\end{example}

\section{Implementation and experiments}
\label{sec:exp}

Simplifications and the new propagator presented in Section~\ref{sec:aggr} have been implemented in \textsc{wasp} \cite{DBLP:conf/lpnmr/AlvianoDLR15}, a modern ASP solver implementing the CDCL algorithm.
Some details of the implementation are given in Section~\ref{sec:implementation}, and an application scenario is presented in Section~\ref{sec:componentassignment}.
\textsc{wasp} already implemented the aggregate propagator described in Section~\ref{sec:propagators}, and therefore it is an ideal framework for evaluating empirically how aggregate set propagators impact the performance of stable model search.
Such an evaluation is reported in Section~\ref{sec:experiment}.

\subsection{Implementation}\label{sec:implementation}

The implementation of the shared aggregate set propagator in \textsc{wasp} introduces a few additional optimizations.
First of all, function OnLiteralTrue is called only if at least one of the conditions at lines~1, 5, 9, and 12 is true.
Actually, only one of these conditions can be true for a literal $\ell$, which is trivially the case for lines~9 and 12, and guaranteed for lines~1 and 5 by performing an additional rewriting during simplifications:
if $(w,\ell)$ and $(w',\overline{\ell})$ occur in a rule $r \in \mathit{rules}^{\sum}(\Pi)$, and $w \geq w'$, then $r$ is replaced by a rule $r'$ of the form (\ref{eq:aggr}) such that $\mathit{id}(r') := \mathit{id}(r)$, $\mathit{elem}(r') := (\mathit{elem}(r) \setminus \{(w,\ell), (w',\overline{\ell})\}) \cup \{(w-w',\ell) \mid w-w' > 0\}$, and $\mathit{bound}(r') := \mathit{bound}(r) - w'$.

The second optimization concerns the implementation of the functions $\mathrm{min\_sum}$, $\mathrm{max\_sum}$, $\mathrm{lower\_bound}$, and $\mathrm{upper\_bound}$, which are heavily used by OnLiteralTrue($X$,$I$,$\ell$).
Instead of computing their values by performing a complete iteration on the sets $\mathit{elem}(X)$ and $\mathit{bounds}(X)$, their values with respect to the current assignment $I$ are stored in local variables, which are updated when new literals are assigned, and during backtracking.

Finally, a third optimization regards the computation of reasons, which are used by the solver only in case some inferred literal is involved in the conflict analysis.
Hence, it is natural to compute reasons only when they are required by the solver.
In order to obtain such a behavior, a \emph{trail} of literals is used by each shared aggregate set propagator to store assigned literals of interest.
Specifically, when OnLiteralTrue($X$,$I$,$\ell$) is invoked, literal $\ell$ is added to the trail of set $X$.
If the conflict analysis requires the reasons of a literal inferred by a shared aggregate set propagator, then an iteration on the trail is sufficient to reconstruct the set of literals shown in function OnLiteralTrue.

\subsection{Component Assignment}\label{sec:componentassignment}

The input of the \emph{Component Assignment} problem is a tuple $(C,p,U,B,\mathcal{I},\mathcal{R})$ defined as follows:
$C$ is a set of components (of a computer);
$p$ is a function mapping each component $c \in C$ to its price;
$U$ is a set of users;
$B$ is a function mapping each user $u \in U$ to an interval of possible expense;
$\mathcal{I}$ is a set of sets of jointly incompatible components;
$\mathcal{R}$ is a set of sets of required components.
The goal is to compute a partial assignment $f : C \pfun U$ of components to users such that the following conditions are satisfied:
for all $u \in U$, $\sum_{c \in C,\ f(c) = u}{p(c)} \in B(u)$ holds (all users expend a permitted amount);
for all sets $\mathcal{I}' \in \mathcal{I}$ there are $c,c' \in \mathcal{I}'$ such that $f(c) \neq f(c')$ (no user can be assigned a set of incompatible components);
for all sets $\mathcal{R}' \in \mathcal{R}$, and for all users $u \in U$, there is a component $c \in \mathcal{R}'$ such that $f(c) = u$ (at least one component in each requirement is assigned to every configuration).
Moreover, the price of each configuration has to be provided in output.

\begin{figure}
\figrule
\begin{align*}
\begin{array}{l}
{\color[rgb]{0,0.7,0}\tt \% \mathrm{guess\ an\ assignment\ of\ components:\ each\ component\ to\ at\ most\ one\ user}}\\
\tt \{assign(C,U) : user(U,\_,\_)\} <= 1 :-\ component(C,\_). \vspace{.3em}\\
{\color[rgb]{0,0.7,0}\tt \% check\ budget}\\
\tt cost(U,COST) :-\ user(U,\_,\_),\ COST = {\color{blue}\#sum}\{P,C : assign(C,U), component(C,P)\}.\\
\tt :-\ user(U,MIN,MAX),\ cost(U,C),\ C < MIN.\\
\tt :-\ user(U,MIN,MAX),\ cost(U,C),\ C > MAX.\vspace{.3em}\\
{\color[rgb]{0,0.7,0}\tt \% check\ incompatibilities\ and\ requirements}\\
\tt :-\ user(U,\_,\_),\ incompatibility(I,\_),\ assign(C,U) : incompatibility(I,C).\\
\tt :-\ user(U,\_,\_),\ requirement(R,\_),\ {\color{blue}not}\ assign(C,U) : requirement(R,C).
\end{array}
\end{align*}
\caption{Tested encoding for the Component Assignment benchmark (ASP-Core-2 syntax).}
\label{fig:encoding}
\figrule
\end{figure}
Figure~\ref{fig:encoding} reports an ASP encoding for Component Assignment.
Recall that this is an abstraction of a real world problem, actually originated in the area of medical informatics dealing with the assignment of patients to operation rooms.
The formulation given here is an excerpt focusing on the main source of inefficiency discovered by analyzing the encoding for the original problem.
Specifically, an aggregate similar to $\tt COST = \#sum\{P,C : assign(C,U), component(C,P)\}$ was the culprit of the inefficiency of ASP solvers, and was subsequently subject to complex optimization (essentially, a minimax preference), which is out of the scope of this paper.

\subsection{Experiments}\label{sec:experiment}

The experiment comprises 147 random instances for increasing number of users (from 2 to 8) and components (from 30 to 50). 
Instances are grounded by \textsc{gringo} \cite{DBLP:conf/lpnmr/GebserKKS11,DBLP:journals/tplp/GebserHKLS15}.
Running time of \textsc{wasp} and \textsc{clasp}~3.3.3 \cite{DBLP:journals/ai/GebserKS12,DBLP:conf/lpnmr/GebserKK0S15} were measured on an Intel Xeon 2.4 GHz with 16 GB of memory.
Time and memory were limited to 20 minutes and 15 GB, respectively.
\textsc{wasp} was tested with aggregate propagators (\textsc{wasp-st}, the standard version of \textsc{wasp}) and with shared aggregate set propagators (\textsc{wasp-sh}, the new prototype).

\begin{table}[b]
    \caption{Solved instances and average running time (in seconds).}
    \label{tab:usercomponents}
    \tabcolsep=0.15cm
    \begin{tabular}{lrrrrrrrrrr}
    \toprule
    \multicolumn{3}{l}{\textsc{Component Assignment}} & \multicolumn{2}{c}{\textbf{\textsc{clasp}}} && \multicolumn{2}{c}{\textbf{\textsc{wasp-st}}} &&    
    \multicolumn{2}{c}{\textbf{\textsc{wasp-sh}}}\\
    \cmidrule{4-5}    \cmidrule{7-8}     \cmidrule{10-11}
    \phantom{xxxxxxxxx} & \textbf{\#users} & \textbf{\#inst} & \textbf{sol.} & \textbf{avg t} && \textbf{sol.} & \textbf{avg t} &&    \textbf{sol.} & \textbf{avg t}\\
    \cmidrule{2-11}
& 2  &21    &    19    &    12.2    &&    19    &    165.8    &&    21    &    3.3\\
& 3    &21&    15    &    11.7    &&    15    &    183.3    &&    17    &    14.1\\
& 4    &21&    8    &    100.1    &&    6    &    386.7    &&    11    &    24.6\\
& 5    &21&    3    &    252.9    &&    4    &    592.8    &&    8    &    63.5\\
& 6    &21&    0    &    -    &&    0    &    -    &&    3    &    189.8\\
& 7    &21&    0    &    -    &&    0    &    -    &&    4    &    529.0\\
& 8    &21&    0    &    -    &&    0    &    -    &&    3    &    464.0\\
\cmidrule{2-11}
& Total    &147 &    45    &    43.7    &&    44    &    240.7    &&    67    &    77.1\\
    \cmidrule{1-11}
    \end{tabular}

    \tabcolsep=0.145cm
    \begin{tabular}{lrrrrrrrrrr}
        \multicolumn{3}{l}{\textsc{ASP Competition}}  & \multicolumn{2}{c}{\textbf{\textsc{clasp}}} && \multicolumn{2}{c}{\textbf{\textsc{wasp-st}}} &&    
        \multicolumn{2}{c}{\textbf{\textsc{wasp-sh}}}\\
        \cmidrule{4-5}    \cmidrule{7-8}     \cmidrule{10-11}
        & \textbf{Problem} &\textbf{\#inst} & \textbf{sol.} & \textbf{avg t} && \textbf{sol.} & \textbf{avg t} &&    \textbf{sol.} & \textbf{avg t}\\
        \cmidrule{2-11}
&        ADF            &    200         & 200      & 23.3         && 120       &    107.3    &&    123         & 110.2    \\
&        Bottle Filling                                                              &100         & 100     &      5.7         && 100       &    5.4          &&  100       &5.6 \\
&        Still Life                                                                     &26             & 6         &     81.3       && 6            &    112.7    &&     6           &31.4\\
&        Still Life with Holes                                                 &120           &     70     &     180.0     &&    90          &    156.9    &&     88          &159.7    \\
        \cmidrule{2-11}
&        Total                                                                        &446        &    376     &    48.7  &&    316      &    89.3  &&    317           &89.5  \\
        \bottomrule
    \end{tabular}

\end{table}
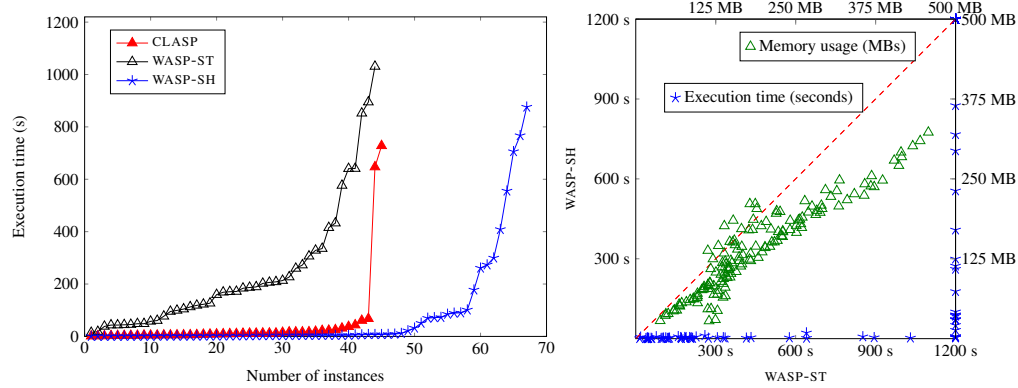
\begin{figure}[t]
    \figrule
    \begin{tikzpicture}[scale=0.7]
    \pgfkeys{%
        /pgf/number format/set thousands separator = {}}
    \begin{axis}[
    scale only axis
    , xlabel={Number of instances}
    , ylabel={Execution time (s)}
    , xmin=0, xmax=70
    , ymin=0, ymax=1220
    , legend style={at={(0.18,0.96)},anchor=north,fill=none}
    , legend columns=1
    , width=0.65\textwidth
    , height=0.45\textwidth
    , ytick={0,200,400,600,800,1000,1200}
    , major tick length=2pt
    , title= {}
    ]
    \addplot [mark size=3pt, color=red, mark=triangle*] [unbounded coords=jump] table[col sep=semicolon, y index=1] {./cactusgenerated.csv}; 
    \addlegendentry{\textsc{clasp}}
    
    \addplot [mark size=3pt, color=black, mark=triangle] [unbounded coords=jump] table[col sep=semicolon, y index=2] {./cactusgenerated.csv}; 
    \addlegendentry{\textsc{wasp-st}}
    
    \addplot [mark size=3pt, color=blue, mark=star] [unbounded coords=jump] table[col sep=semicolon, y index=3] {./cactusgenerated.csv}; 
    \addlegendentry{\textsc{wasp-sh}}    
    \end{axis}
    \end{tikzpicture}%
    \hfill
    \begin{tikzpicture}[scale=0.7]
    \pgfkeys{%
        /pgf/number format/set thousands separator = {}}
    \begin{axis}[
    scale only axis
    , legend style={at={(0.4,0.8)}, anchor=north, align=left}
    , legend cell align=left
    , xlabel={\textsc{wasp-st}}
    , ylabel={\textsc{wasp-sh}}
    , width=0.45\textwidth
    , height=0.45\textwidth
    , xmin=0, xmax=1200
    , ymin=0, ymax=1200
    , xtick={300,600,900,1200}
    , xticklabels={300 s,600 s,900 s,1200 s}
    , ytick={300,600,900,1200}
    , yticklabels={300 s,600 s,900 s,1200 s}
    , major tick length=2pt
    ]
    \addplot [mark size=3pt, only marks, color=blue, mark=star] [unbounded coords=jump] table[col sep=semicolon, x index=2, y index=3] {scatter-performance.csv};     
    \addlegendentry{Execution time (seconds)}
    \addplot [color=red, dashed] [unbounded coords=jump] table[col sep=semicolon, x index=0, y index=0] {scatter-performance.csv}; 
    \end{axis}
    \begin{axis}[
    scale only axis
    , legend style={at={(0.6,0.96)}, anchor=north, align=left}
    , legend cell align=left
    , width=0.45\textwidth
    , height=0.45\textwidth
    , xmin=0, xmax=500
    , ymin=0, ymax=500
    , xtick={125,250,375,500}
    , xticklabels={125 MB,250 MB,375 MB,500 MB}
    , ytick={125,250,375,500}
    , yticklabels={125 MB,250 MB,375 MB,500 MB}
    , axis y line*=right
    , axis x line*=top
    , major tick length=3pt
    ]
    \addplot [mark size=3pt, only marks, color=green!50!black, mark=triangle] [unbounded coords=jump] table[col sep=semicolon, x index=1, y index=2] {./scatter-mem-usercomp.csv};     
    \addlegendentry{Memory usage (MBs)}
    \addplot [color=red, dashed] [unbounded coords=jump] table[col sep=semicolon, x index=0, y index=0] {./scatter-mem-usercomp.csv}; 
    \end{axis}
    \end{tikzpicture}

    \caption{Component Assignment: number of solved instances within a given execution time limit (cactus plot on the left), and instance-by-instance comparison of execution time and memory consumption of the two versions of \textsc{wasp} (scatter plot on the right).}\label{fig:cactusgenerated}    
    \figrule
\end{figure}
The results are given in Table~\ref{tab:usercomponents}.
Even if \textsc{clasp} is faster than \textsc{wasp-st}, solving one more instance overall, both solvers cannot terminate on instances with 6 or more users.
A better performance is reached by \textsc{wasp-sh}, which can solve some instances up to 8 users, and is in general faster than \textsc{clasp} and \textsc{wasp-st}, as confirmed by the cactus plot shown in Figure~\ref{fig:cactusgenerated}.
The figure also shows an instance-by-instance comparison of the performance of the two tested versions of \textsc{wasp}.
Concerning execution time, it can be observed that all the points are below the diagonal, meaning that the propagator always provides a performance gain for the tested instances.
Often, shared aggregate set propagators also reduce the memory footprint of \textsc{wasp}.

A few additional testcases are obtained from instances of ASP Competitions.
The aim of this benchmark is to verify the absence of overhead when shared aggregate set propagators are applied to encodings that, differently from the program in Figure~\ref{fig:encoding}, do not have assignments over aggregate sets with nonuniform weights.
In fact, the benchmarks comprise instances that include assignments over count aggregates, that are, Abstract Dialectical Framework (ADF), Bottle Filling, Connected Maximum-density Still Life (Still Life), and Connected Maximum-density Still Life with Holes (Still Life with Holes).
Results are given in Table~\ref{tab:usercomponents}, where the fact that \textsc{clasp} is in general faster than \textsc{wasp-st} is confirmed by a gap of 60 instances, mainly due to instances of ADF.
Comparing \textsc{wasp-st} and \textsc{wasp-sh}, instead, it can be observed that overall there is no overhead on using the shared aggregate set propagators, and actually there is a slight performance gain.
The performance gain depends on the number of aggregates sharing the same aggregate set, as in fact for instances of ADF \textsc{wasp-st} uses 900 aggregate propagators on average, while \textsc{wasp-sh} only requires 20 shared aggregate set propagators on average.
For Still Life, there is no difference in terms of solved instances, but \textsc{wasp-sh} has a clear advantage over \textsc{wasp-st} in terms of running time, again due to the number of aggregates sharing the same aggregate set.
All in all, there is no overhead on using the shared aggregate set propagators also when the input program does not contain assignments over aggregate sets with nonuniform weights.

\section{Related work}
\label{sec:rel}

Rules of the form (\ref{eq:aggr}) were introduced in \textsc{smodels} \cite{DBLP:journals/ai/SimonsNS02}, and called \emph{basic constraint rules}.
From a computational point of view, there are two mainstream approaches to evaluate ASP programs with aggregates, here referred to as \textit{translation-based} and \textit{propagator-based}.

The first approach aims at rewriting aggregates in terms of other constructs.
For example, the similarities between aggregates and pseudo-Boolean constraints were used to adapt to the ASP setting some compilations of pseudo-Boolean constraints into clauses \cite{DBLP:conf/sara/AavaniMT13}.
Among them, there are adder circuits and binary decision diagrams \cite{DBLP:journals/jair/AbioNORM12,DBLP:journals/jsat/EenS06}, sorting networks and watchdogs \cite{DBLP:conf/sat/BailleuxBR09}.
ASP versions of these translations are implemented in \textsc{lp2sat} and \textsc{lp2normal} \cite{DBLP:conf/jelia/BomansonGJ14,DBLP:conf/lpnmr/BomansonJ13}, where the first solver outputs CNF formulas and the second solver outputs rules of the form (\ref{eq:rule}).
A translation-based approach is also implemented in \textsc{cmodels} \cite{DBLP:conf/lpnmr/LierlerM04,DBLP:journals/jar/GiunchigliaLM06,DBLP:journals/amai/GiunchigliaLM08}, which maps aggregates to nested logic programs \cite{DBLP:journals/tplp/FerrarisL05}.

The concept of propagator was introduced in Satisfiability Modulo Theories \cite{DBLP:journals/jacm/NieuwenhuisOT06,DBLP:conf/cp/AbioS12}, and used in ASP to handle unfounded sets \cite{DBLP:conf/lpnmr/BomansonGJKS15}, support \cite{DBLP:conf/ijcai/AlvianoD16}, and also to avoid the instantiation of constraints \cite{DBLP:journals/tplp/CuteriDRS17}.
Moreover, some extensions of ASP, such as CASP \cite{DBLP:conf/iclp/BaseliceBG05}, were implemented by adding propagators to CDCL solvers \cite{DBLP:journals/tplp/BanbaraKOS17}.
Propagators are also the basis of the implementation of the solver \textsc{idp}~\cite{DBLP:journals/tplp/BruynoogheB0CPJ15}.

Concerning aggregates, \textsc{dlv} \cite{DBLP:conf/datalog/AlvianoFLPPT10,DBLP:conf/lpnmr/AlvianoCDFLPRVZ17} implements ad-hoc techniques to evaluate programs with aggregates, among them a hashmap to compactly represent shared aggregate sets in the input program \cite{DBLP:journals/tplp/FaberPLDI08}.
Differently from the technique proposed in this paper, \textsc{dlv} could identify shared aggregate sets only at a symbolic level, before the grounding process.
Actually, the aim of \textsc{dlv} was to speedup the grounding phase by not instantiate several times the same aggregation set, and no further advantage was obtained during the solving phase (if not a reduced memory footprint).
Specifically, redundant computation during propagation was not eliminated in \textsc{dlv}.
Another difference with the propagator introduced in this paper is that the stable model search algorithm implemented in \textsc{dlv} is much simpler, and does not support constraint learning and non-chronological backtracking, key features of modern solvers.
The computation of reasons was supported in some version of \textsc{dlv}, but only for backjumping rather than for constraint learning \cite{DBLP:journals/fuin/FaberLMR11};
moreover, in \textsc{dlv} reasons are computed as soon as a literal is inferred, rather than postponed to when they are used.

The state-of-the-art solver \textsc{clasp} implements a hybrid approach for handling programs with aggregates \cite{DBLP:conf/iclp/GebserKKS09}, where aggregates involving few literals are compiled into normal rules;
the threshold on the number of literals can be configured from the command-line.
The aggregate propagator introduced in \textsc{clasp} takes advantage of a trail of literals in order to postpone the computation of reasons.
The same trail was used in the aggregate propagator of \textsc{wasp}.
Shared aggregate sets generalize this propagator by compactly representing several aggregate propagators differing only on their bounds.

Aggregates in ASP can be interpreted according to several semantics \cite{DBLP:journals/tocl/Ferraris11,DBLP:journals/ai/FaberPL11,DBLP:journals/tplp/PelovDB07,DBLP:journals/tplp/SonP07,DBLP:journals/tplp/GelfondZ14}, which however agree for programs with non-recursive aggregates.
This is the main reason for inhibiting cycles over aggregates in Section~\ref{sec:syntax}.
In fact, recursive aggregates require specific techniques in addition to the propagators presented in Section~\ref{sec:propagators}, as for example unfounded set detection for the semantics by \citeANP{DBLP:journals/tocl/Ferraris11}.
Such techniques can still be used with the propagator introduced in Section~\ref{sec:aggr}, even in the non-convex case (being the complexity boundary; \citeNP{DBLP:conf/lpnmr/AlvianoF13}) thanks to a rewriting into monotone aggregates \cite{DBLP:journals/tplp/AlvianoFG15,DBLP:conf/ijcai/AlvianoFG16,DBLP:journals/fuin/Alviano16}.
Interestingly, recursive aggregates can be eliminated for the semantics by \citeANP{DBLP:journals/tplp/GelfondZ14} by means of \emph{polynomial, faithful and modular} \cite{DBLP:journals/jancl/Janhunen06} translation functions \cite{DBLP:journals/tplp/AlvianoL15,DBLP:conf/ijcai/AlvianoL16}.
Finally, queries over (super-coherent) ASP programs \cite{DBLP:journals/tplp/AlvianoFW14} with aggregates can be optimized by magic sets \cite{DBLP:conf/lpnmr/AlvianoGL11,DBLP:journals/aicom/AlvianoF11,DBLP:journals/ai/AlvianoFGL12}, and the propagator introduced in this paper can be used also in presence of such an optimization.

Aggregates are also used by algorithms for computing optimal answer sets that are based on unsatisfiable core analysis.
Among them are \textsc{oll} \cite{DBLP:conf/iclp/AndresKMS12} and \textsc{one} \cite{DBLP:conf/ijcai/AlvianoDR15,DBLP:journals/tplp/AlvianoD16,DBLP:conf/ijcai/AlvianoD17}.
\textsc{oll} introduces several aggregates that can be compactly represented by the propagator introduced in this paper, so as to obtain the same behavior of \textsc{one}.

\section{Conclusion}
\label{sec:concl}

ASP syntax allows to bind object variables to the result of an aggregation, frequently a sum.
Stable models of these programs may have constants not occurring in the input, but being among the possible results of such aggregations.
For computing these stable models, mainstream ASP systems have to instantiate such aggregations for all possible interpretations of the undefined literals occurring in aggregate sets involved in an assignment.
All these ground aggregates cause redundant computation of the solver.
Shared aggregate set propagators are conceived to overcome such an inefficiency, and can be applied directly to the output of a grounder, not relying on any knowledge on the origin of the ground aggregates.
The potential performance gain on \textsc{wasp} is of orders of magnitude, and also memory footprint is reduced.
Finally, no overhead is introduced when there are few shared aggregate sets in the input program.

\section{Acknowledgments}

Mario Alviano has been partially supported by the POR CALABRIA FESR 2014-2020 project ``DLV Large Scale'' (CUP J28C17000220006), by the EU H2020 PON I\&C 2014-2020 project ``S2BDW'' (CUP B28I17000250008), and by GNCS-INdAM.
Carmine Dodaro has been partially by the Grant FFABR, ``Fondo per il finanziamento delle attivit\`a base di
ricerca, comma 295 della Legge di Bilancio 2017 (Legge 232/2016)''.

\bibliographystyle{acmtrans}

\label{lastpage}

\end{document}